\documentclass[twoside,leqno,twocolumn]{article}
\usepackage{ltexpprt}

\usepackage{amsmath,bm,graphicx,amsfonts}
\usepackage{caption}
\usepackage{subcaption}
\usepackage{arydshln}

\makeatletter
\def\hlinewd#1{%
  \noalign{\ifnum0=`}\fi\hrule \@height #1 \futurelet
   \reserved@a\@xhline}
\makeatother

\newcommand{\vect}[1]{\boldsymbol{\mathbf{#1}}}

\begin{document}

\makeatletter
\def\footnoterule{\kern-8\p@
  \hrule \@width 2in \kern 3\p@} 
\makeatother

\title{\Large Graph Model Selection via Random Walks\thanks{This work was sponsored
by the Defense Advanced Research Projects Agency under Air Force Contract FA8721-05-C-0002. Opinions, interpretations, conclusions, and recommendations are those of the authors and are not necessarily endorsed by the United States Government.}}
\date{}
%
\author{Lin Li, William M. Campbell, Rajmonda S. Caceres\\
MIT Lincoln Laboratory\\
\texttt{ \{lin.li, wcampbell, rajmonda.caceres\}@ll.mit.edu}}


%
%
%

%
\maketitle
%
%

%
\begin{abstract}
%
%
In this paper, we present a novel approach based on the random walk process for finding meaningful representations of a graph model. Our approach leverages the transient behavior of many short random walks with novel initialization mechanisms to generate model discriminative features. These features are able to capture a more comprehensive structural signature of the underlying graph model. The resulting representation is invariant to both node permutation and the size of the graph, allowing direct comparison between large classes of graphs. We test our approach on two challenging model selection problems: the discrimination in the sparse regime of an Erd\"{o}s-Renyi model from a stochastic block model and the planted clique problem. Our representation approach achieves performance that closely matches known theoretical limits in addition to being computationally simple and scalable to large graphs.

\end{abstract}
%


%

\section{Introduction}\label{sec:intro}

%
%
Graphs are important data abstractions that allow us to analyze complex relational patterns in many application domains, such as social networks, information networks and protein networks. An active area of research focuses on theoretical models that define the generative mechanism of a graph. The mapping of an observed graph instance to a
model allows us to apply the theoretical knowledge we have about the model and to make precise
claims about the underlying structure of the data. Yet given the complexity and inherent noise in real
datasets, it is still very challenging to identify the best model for a given observed graph. We
discuss an approach for graph model selection that leverages embeddings of graphs in high dimensional feature space. In addition to gaining insight about graph structure, feature representations of graphs allow for the application of traditional tools in machine learning, many of which require continuous vector representations as input. 

In this paper, we introduce {\it Walk2Vec}, which uses random walks to characterize and compare structural properties of the underlying graphs. Our approach is based on the following intuition: structurally different graphs are likely to exhibit different random walk characteristics. For example, 
some graphs diffuse faster across nodes than others. 
Inspired by the intuition that diffusion properties reveal discriminating features for comparing graphs, we propose a flexible framework for mapping graphs or any substructures into an Euclidean feature space. Our approach has several important characteristics. It is invariant to both node labeling and the size of the graph, making it more appropriate for realistic applications. It leverages different diffusion instantiations on the same graph to capture a richer profile of graph structure. Finally, our approach maintains its performance robustness while being computationally efficient.

Various efforts in the literature ~\cite{Perozzi2014,Yanardag2015,Narayanan2016,Grover2016}  have focused on learning vector representations of nodes or subgraphs in the graph. 
Our approach differs from the above approaches by offering a unifying and flexible framework that can learn a feature representation of a single node in the graph, a subgraph or the whole graph. At the same time, we learn representations that are both node label and graph size invariant. This is an important characteristic considering that in many practical settings, we need to compare graphs of different sizes or graphs with different node labelings. However, in many such settings differences in graph size and node labelings do not necessarily imply different graph structural properties. Finally, our approach follows a different and novel way of capturing the rich and diverse structural patterns in a graph. By using and combining different initialization mechanisms for random walks on the same graph, we generate a more comprehensive structural signature of the graph.

A much closely related line of work focuses on mapping the graph as a whole into a topological or spectral space, often in the context of the model selection or graph classification task. Many features are considered to represent the graph, from various density and path features~\cite{caceres2016,Airoldi2011}, to distributions of frequent subgraphs or graphlets~\cite{Middendorf2004,Janssen2012}, to spectral features of graph matrices~\cite{Zhu2005,Fay2011,Takahashi2012}. 
Our approach leverages random walk based features to capture the essential structure of the graph. As we demonstrate in Sec.~\ref{sec:class}, for two extensively-studied model selection cases, our approach achieves classification performance very close to known theoretical results. Furthermore, when compared to topological embeddings of graphs~\cite{caceres2016,Airoldi2011}, we show that our random walk features lead to much more tightly clustered embeddings of similar graph instances.

Finally, in the graph kernel literature~\cite{Vishwanathan2010,Bach2008,Borgwardt2005,Gartner2003,Kashima2003}, similarity between two graphs is defined as a function of the number of matching random walks. A major focus in this literature is the reduction of computational complexity with various results giving polynomial time algorithms~\cite{Vishwanathan2010}.  
An important feature of our algorithm is its computational efficiency, making it scalable to big data settings.

\section{Problem and General Approach}\label{sec:prob}

%
%
\newcommand{\cc}{{\cal C}}
\newcommand{\cp}{{\cal P}}

Our motivating problem is graph classification. Let $G=(V,E)$ denote a connected graph with node set $V$ and edge set $E$. Let $\cc_1$ and
$\cc_2$ be two distinct random graph models. The goal is to classify a graph $G$ as being drawn from one
of these distributions $G\sim \cc_1$ or $G\sim\cc_2$. Although the problem can be more general, in this paper, we focus on the problem of selecting the best model given a set of generative graph models. 

Given an input graph $G$, our general
approach is to use random walks with different initial distributions to map the graph $G$ into a
Euclidean space $\phi(G)$, and then use standard machine learning
methods to perform supervised learning and assign the graph point $\phi(G)$ to its closest generative model.

\begin{figure}[h]
\vspace{-0.05in}
\centering
\includegraphics[clip,width=0.8\linewidth]{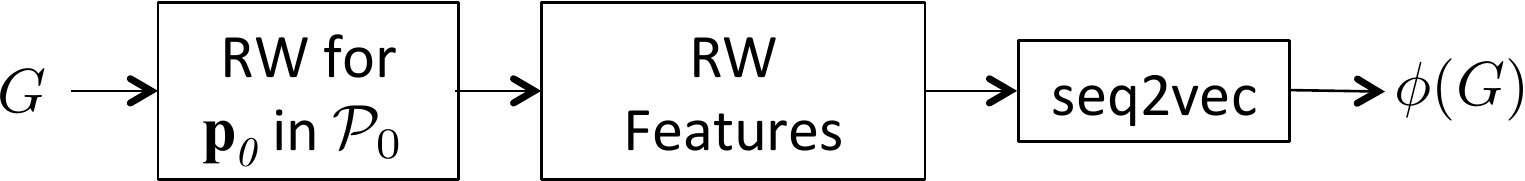}
\caption{System for {\it Walk2Vec}}\label{fig:w2v}
\vspace{-0.1in}
\end{figure}

The proposed approach, {\it Walk2Vec}, for random walk features and
mapping is shown in Figure~\ref{fig:w2v}. The first stage is to generate random walks on the graph $G$. Let $\mathcal{P}_0$ be a set of initial distributions over the nodes $\{1, \cdots, n\}$ where $n = |V|$. For each distribution $\vect{p}_0 \in \mathcal{P}_0$, perform $\tau$ random walk
steps on $G$ and correlate the $\tau$ steps to produce a random walk feature
vector $r(\vect{p}_0)$ (see Section ~\ref{sec:rw}). Then the resulting sequence of random walk features $\{r(\vect{p}_0)\}$ for all $\vect{p}_0 \in \mathcal{P}_0$ is then mapped into a single vector $\phi(G)$ in Euclidean space using the {\rm seq2vec} operation as shown in Figure~\ref{fig:w2v}.
%
%

A desirable property for $\phi(G)$ is that it is
invariant under graph isomorphism (permutation of the node
labels). In this paper, we achieve this property with two different
strategies. The first one is via careful selection of the initial distributions in
$\cp_0$ to produce invariant random walk features $r(\vect{p}_0)$; each feature vector $r(\vect{p}_0)$ is independent of the node labels. In this case, the {\rm seq2vec} can simply be a stacking operator and the resulting vector representation $\phi(G)$ of the graph $G$ is also independent of the node labels and thus permutation invariant (see Section~\ref{sec:rw_em}). 

The second strategy is to generate random walk features that are localized to each of the nodes in the graph $G$. For example, the random walks are initialized from a single node. In this case, each feature vector $r(\vect{p}^i_0)$ is associated with a node label $i$ where $i = 1, \cdots, n$. To aggregate the sequence of vectors $\{r(\vect{p}_0^i)\}_{i=1}^n$, the {\rm seq2vec} operation performs two steps: sparse coding and pooling. Sparse coding extracts high-level features from each $r(\vect{p}_0^i)$. The pooling step then combines these high-level features together and outputs a vector representation $\phi(G)$ of the graph $G$. The pooling operation has been used extensively in machine learning to achieve a better representation for classification. Additionally, in our case, pooling also provides a form of invariance under permutation of node labels (see Section~\ref{sec:sc_em}).

Several generalizations of our problem and approach should be
mentioned.  First, although we consider graph model selection, the
methods can be easily extended to graph classification or subgraph classification\footnote{Subgraph classification can be achieved by using a
fixed local sampling method, e.g., the ego-net of a node or the subgraph induced by a breadth-first-search.}. Second, the mapping $\phi(\cdot)$ is robust
and can be used for clustering, regression, etc.

\section{Random walk features and mapping}\label{sec:g2v}

%
%
%
Random walks have led to many important graph properties~\cite{lovasz1993random}, i.e., hitting time, mixing time, commute time, etc., and can also provide us with considerable insight into the structure of the graph.  

%
\subsection{Random-Walk Graph Features}\label{sec:rw}

The random walk process on a graph $G=(V,E)$ can be represented as follows. Let $\vect{A}$ be the adjacency matrix associated with $G$: $A_{ij} = 1$ if nodes $(i,j) \in E$ and $A_{ij} = 0$ otherwise. The probability of moving from the current node to the neighbor is given by the transition matrix 
$
{\vect W} = {\vect D}^{-1}{\vect A} 
$,
where ${\vect D}$ is the diagonal matrix of the degrees, i.e., $D_{ii} = d_i = \sum_{j}A_{ij}$. 
Suppose that the initial node is drawn from some initial probability distribution ${\vect p}_0$. The probability distribution after $t$ steps of random walk is given by
\begin{align}\label{walk_eqn}
{\vect p}_t = {\vect W}^T{\vect p}_{t-1} = ({\vect W}^T)^t{\vect p}_0\enspace.
\end{align}
For simplicity, we only consider the unweighted graphs. However, our approach can be easily extended to weighted graphs, where ${A}_{ij}~\in~\mathbb{R}_+$.

A fundamental property of a random walk on a connected, undirected and non-bipartite graph is that asymptotically, the probability distribution converges to a unique stationary distribution $\vect{\omega}$ where
$
\omega_i = {d_i}/{\sum_{k}d_k}
$.
That is, asymptotically the probability of being at node $i$ only depends on the degree of node $i$ and not on the initial node.

Let us consider random walks on $G$ with length~$\tau$. Note that: 1) the number of random walk steps $\tau$ must be sufficiently long to capture enough information about the graph; and 2) the probability distribution ${\vect p}_t$ is biased towards the high-degree nodes as the number of random walk steps $\tau$ increases. 

To generate the random walk feature on a graph, we now introduce a pair-wise distance matrix ${\vect M}$ between probability distributions $\{{\vect p}_0, {\vect p}_1, \cdots, {\vect p}_\tau\}$ up to step $\tau$: 
\begin{align}\label{distance}
M_{s t} = \|{\vect D}^{-\frac{1}{2}}{\vect p}_s - {\vect D}^{-\frac{1}{2}}{\vect p}_{t}\|_2 \enspace.
\end{align}
where $0\le s,t\le \tau$.
Elements in $\vect{M}$ capture the temporal changes between various steps of the random walk. This distance can also be seen as the $L^2$ distance between the two probabilities ${\vect p}_s$ and ${\vect p}_t$ with respect to the stationary distribution ${\vect \omega}$~\cite{aldous2002reversible}.
This measure has also been used in~\cite{pons2005computing} to compute distance between nodes using random walks on a graph for community detection and has shown good performance.

Given an initial distribution ${\vect p}_0$, the random walk feature on graph $G$ is defined by the function $r: \mathcal{P}_0 \rightarrow \mathbb{R}_+^d$ with $d = \frac{\tau^2+\tau}{2}$ and
\begin{align}\label{rwfeat}
r({\vect p}_0) = \operatorname{triu}({\vect M})
\end{align}
where $\vect{M}$ is defined in Eqn. \eqref{distance} and $\operatorname{triu}(\cdot)$ returns the upper triangular elements of the matrix. Note that $r({\vect p}_0)$ is independent of the graph size.

In this paper, we use the $L^2$ distance for computing the random walk feature. However, it is trivial to fit our framework to other distance metrics between two probabilities, such as total variation distance~\cite{clarkson1933definitions} and symmetric Kullback-Leibler (KL) divergence~\cite{kullback1951information, johnson2001symmetrizing}.  One can also use the similarity matrix to generate the random walk feature. An example is to compute the following similarity matrix ${\vect S}$:
\begin{align}\label{similarity}
S_{s t} = \frac{{\vect p}_s^T{\vect D}^{-1}{\vect p}_{t}}{\|{\vect D}^{-\frac{1}{2}}{\vect p}_s\|_2\|{\vect D}^{-\frac{1}{2}}{\vect p}_{t}\|_2} \enspace.
\end{align}
where $0\le s,t\le \tau$.
and replace ${\vect M}$ with ${\vect S}$ in Eqn. \eqref{rwfeat}.


%
%
%
\subsection{Walk2Vec}\label{sec:rw_em}

To map the input graph into a Euclidean space, we now restrict the random walk features $\{r(\vect{p}_0)\}$ to those that are invariant to node label. This property is addressed by choosing an appropriate initial probability distribution $\vect{p}_0$ for the random walks.

\begin{lemma}\label{lem:condition}
Let ${\vect A}$ be the adjacency matrix of the graph $G=(V,E)$ and let ${\vect p}_0 = \frac{\vect{g}(A)}{\|\vect{g}(A)\|_1}$ be an initial distribution on $G$, where $\vect{g}$ is a nonnegative function of the adjacency matrix. If for any permuation matrix $\vect{\Pi}$ of compatible dimension,
\begin{align}\label{condition}
\vect{g}(\vect{\Pi}A\vect{\Pi}^T)=\vect{\Pi}\vect{g}(A)\enspace,
\end{align}
then the distance matrix ${\vect M}$ defined in Eqn.~(\ref{rwfeat}) is invariant under node permutation.  
\end{lemma}
\begin{proof}
Suppose ${\vect A}$ and $\tilde{\vect A}$ are two adjacency matrices, where there exists a permutation ${\vect \Pi}$ such that $\tilde{\vect A} = {\vect \Pi}{\vect A}{\vect \Pi}^T$, i.e., ${\vect A}$ and $\tilde{\vect A}$ represent the same graph with a node permutation. Then their associated degree matrices and transition matrices follow the relation $\tilde{\vect D} = {\vect \Pi}{\vect D}{\vect \Pi}^T$ and $\tilde{\vect W} = {\vect \Pi}{\vect W}_1{\vect \Pi}^T$, respectively. And the associated random walk distributions ${\vect p}_{t|{\vect A}}$ and ${\vect p}_{t|\tilde{\vect A}}$ at step $t$ are 
\begin{align}\label{perm}
 {\vect p}_{t|\tilde{\vect A}} &= (\tilde{\vect W}^T)^t{\vect p}_{0|\tilde{\vect A}} = {\vect \Pi}({\vect W}^T)^t{\vect \Pi}^T{\vect p}_{0|\tilde{\vect A}}\enspace,
\end{align}
where ${\vect p}_{0|{\vect A}}:=\frac{{\vect g}({\vect A})}{\|{\vect g}({\vect A})\|_1}$ and ${\vect p}_{0|\tilde{\vect A}}:=\frac{{\vect g}(\tilde{\vect A})}{\|{\vect g}(\tilde{\vect A})\|_1}$. Then it follows from Eqn. \eqref{condition} that ${\vect p}_{0|\tilde{A}} = {\vect \Pi}{\vect p}_{0|{\vect A}}$. Subsequently, Eqn. \eqref{perm} becomes
\begin{align}
{\vect p}_{t|\tilde{\vect A}} = {\vect \Pi}({\vect W}^T)^t{\vect p}_{0|{\vect A}} = {\vect \Pi}{\vect p}_{t|{\vect A}}\enspace.
\end{align}
Then it is easy to check that the distance matrix ${\vect M}$ is invariant under the permutation ${\vect \Pi}$.
\end{proof}  

Examples of $\vect{p}_0$ that satisfy the condition in Lemma~\ref{lem:condition} include the uniform distribution, normalized centrality vector, normalized local clustering coefficient, etc.  
Additionally, $g(i)$ can also be a function that selects one (or a subset) of the nodes, i.e., the node with the highest centrality value or the highest clustering coefficient.
Although there are many ways of selecting the initial distribution $\vect{p}_0$ that is permutation invariant. In Section~\ref{sec:expt}, we will show one way of selecting the set of initial distributions that gives good results. 

Suppose that $\mathcal{P}_0$ is the set of initial distributions over the graph; each ${\vect p}_0 \in \mathcal{P}_0$ leads to a permutation-invariant feature vector ${\vect x}=r({\vect p}_0)$. Then a common approach to construct the mapping $\phi(G)$ is by stacking the sequence of random walk features $\{r(\vect{p}_0)\}$ into a single vector; see Section~\ref{sec:expt} for an example.



%
%
\newcommand{\bD}{\vect{D}}
\newcommand{\bx}{\vect{x}}
\newcommand{\by}{\vect{y}}
\newcommand{\argmin}{\operatornamewithlimits{argmin}}
\newcommand{\argmax}{\operatornamewithlimits{argmax}}

\subsection{Walk2Vec-SC}\label{sec:sc_em}
As discussed previously, one way to generate a mapping $\phi(G)$ of the graph $G$ is to restrict each initial distribution $\vect{p}_0 \in \mathcal{P}_0$ such that the random walk feature $r(\vect{p}_0)$ is invariant to permutation of node labels. Alternatively, one can choose a set of initial distributions that are localized to the nodes of $G$. The permutation invariant property of the mapping $\phi(G)$ can be achieved by sparse coding (SC) the localized random walk features, followed by a pooling operation. 

To extract localized random walk features, we use a set of initial probability distributions
${\cal P}_0 = \{\vect{p}_0^1, \cdots, \vect{p}_0^n\}$; each $\vect{p}_0^i$ is localized in the graph to node~$i$.  Several examples are: 1) a delta distribution which is one
only on the $i$th node, $\vect{p}_0^i=\vect{e}_i$ where $\vect{e}_i$
is the $i$th column of the identity matrix, or 2) a uniform
distribution on the ego-net of the $i$th node.  We then find random
walk features as in~\eqref{rwfeat} for each $\vect{p}_0^i \in \mathcal{P}_0$ to
obtain a sequence of feature vectors $\vect{x}_i=r(\vect{p}_0^i)$.

If we combined the vectors $\vect{x}_i$ using a simple function such
as averaging $\bar{\vect x} = \frac{1}{n}\sum_{i=1}^n\vect{x}_i$, then the output would be a single vector that is also permutation invariant. However,
significant information would be lost in the averaged vector $\bar{\vect x}$.  A better way to approach this problem is to find a compact high-level representation for each of the feature vectors ${\vect x}_i$ using sparse coding \cite{yang2009linear, Gwon16} and then aggregate these high-level representations together using pooling. 

We train a matrix ${\bf D}$, the dictionary, that is used to represent
an input $\bx_i$ as ${\bf D}\by_i\approx\bx_i$ where $\by_i$ is a
sparse vector. The dictionary is overcomplete; that is, redundant
information is included to allow a sparse solution for $\by_i$.  For a
known ${\bf D}$, we use the LASSO criterion~\cite{tibshirani1996regression} for sparse
coding
\begin{equation}\label{eqn:lasso}
\by_i = \argmin_{\hat{\by}_i} \frac{1}{2}\|\bD{\hat{\by}_i}-\bx_i\|_2^2 + \lambda_1\|{\hat{\by}_i}\|_1.
\end{equation}
The $\ell_1$ penalty in~\eqref{eqn:lasso} encourages sparsity in $\by_i$.  Dictionary learning is via the methods
in~\cite{mairal2010online}, and LASSO is solved with the LARS
algorithm~\cite{efron2004least}, both in the SPAMS software
package. Intuitively, the dictionary atoms (columns of $\bD$)
represent different random walks, and the sparse coordinates $\by_i$ are
the atoms seen at node~$i$.

The mapping $\phi(G)$ is then computed by pooling the sparse coded
vectors $\by_i$ across all $i$.  Pooling is either average pooling
$\phi(G) = \frac{1}{n}\sum_{i=1}^{n}\by_i$
or max pooling
$ \phi(G)_j = \argmax_{i} y_{i,j}$,
where $y_{i,j}$ is the $j$th component of $\by_i$ and
$\phi(G)_j$ is the $j$th component of $\phi(G)$. Sparse coding and
pooling perform the {\it seq2vec} operation in Figure~\ref{fig:w2v}.
Pooling has the property that it creates a permutation invariant
mapping of the graph.  That is, reordering the node
labels will not change the mapping~$\phi(G)$.

\section{Experiment Setup}\label{sec:expt}
For Walk2Vec, we first need to select a set of initial distributions
that satisfy the condition in Lemma~\ref{lem:condition}. In the
experiment, we consider a set of four initial distributions
$\mathcal{P}_0 = \{{\vect p}^{\rm max}_0, {\vect p}^{\rm min}_0,
{\vect p}^{\rm median}_0, {\vect p}^{\rm mean}_0\}$. Specifically,
${\vect p}^{\rm max}_0$ (or ${\vect p}^{\rm min}_0$) corresponds to a
delta distribution ${\vect e}_i$, where $i$ is the index of the node
that has the maximum (or minimum) node degree in the graph. Similarly, ${\vect p}^{\rm median}_0$ (or ${\vect p}^{\rm mean}_0$) corresponds to a delta distribution ${\vect e}_i$, where $i$ is the index of the
node whose degree is the closest to the median (or mean) node degree
in the graph. In other words, the random walks on $G$ are initialized
from each of the four above-mentioned type of nodes. Note that to satisfy the
condition in Eqn.~\eqref{condition}, the selected nodes also need to
be unique. In the case that there exist more than one maximum (or
minimum) degree node, we pick the one that has the maximum (or
minimum) PageRank in the graph. A similar strategy is used in
selecting the median and mean degree node in the graph. The resulting
representation of the graph is given by stacking all the feature
vectors $\phi(G) = [r({\vect p}^{\rm max}_0)^T, r({\vect p}^{\rm
min}_0)^T ,r({\vect p}^{\rm median}_0)^T,r({\vect p}^{\rm
mean}_0)^T]^T$.

The setup for the Walk2Vec-SC system is as follows. For each graph, we
use the set $\mathcal{P}_0 = \{{\vect p}^{1}_0, \cdots, {\vect
p}^n_0\}$, where ${\vect p}_0^i = {\vect e}_i$, to initiate the random
walk process. This leads to a sequence of feature vectors ${\vect x}_i
= r({\vect p}^i_0)$. Given a training set, a dictionary of $100$ atoms
is trained using the SPAMS tool for Python.  The dictionary is used to
convert the random features into sparse vectors for each node
using~\eqref{eqn:lasso}, and we set $\lambda_1=0.15$ for all the
experiments. After computing sparse vectors, for each graph the
mapping $\phi(G)$ is found using pooling.


Furthermore, we train a random forest classifier~\cite{Breiman:2001:RF:570181.570182} and use the learned model to classify a collection of unlabeled graph instances.  
All random forest classifiers are trained with $100$ decision trees.


\section{Graph Model Selection: Case Studies}\label{sec:class}

%
%
To validate our new approach, we apply Walk2Vec and Walk2Vec-SC to two
graph model selection problems: 1)~Erd\"{o}s-Renyi vs. stochastic
block model, and 2)~planted clique problem.

\subsection{Erd\"{o}s-Renyi vs. Stochastic Block Model}\label{sec:sbm}
The problem of distinguishing between an
Erd\"{o}s-Renyi (ER) graph~\cite{Erdos60} and a stochastic
block model (SBM)~\cite{Snijder97} is as follows. Let $p$ be the connection probability for the ER graph. Consider the case of a SBM graph with two communities (blocks) of equal size $n/2$ where $n=|V|$. The cross-community probability is $p_{\rm out}$ and the within-community probability is $p_{\rm in}$; the density is given by $(p_{\rm in}+p_{\rm out})/2$. As the
difference $p_{\rm in} - p_{\rm out} >0$ becomes smaller, the SBM graphs become
harder to distinguish from ER graphs of the same density. Let $\delta = p_{\rm in} - p_{\rm out}$. For graphs of the same density, i.e., $p = (p_{\rm in}+p_{\rm out})/2$, the theoretical
limit~\cite{nadakuditi2012graph} for discriminating the two models is 
\begin{align}
\delta_{\rm crit.} = 2\sqrt{\frac{p}{n}}\enspace.
\end{align}
This limit offers a precise mechanism for quantifying the robustness of our two graph mapping algorithms. 

For each pair of parameters $(p, \delta)$, we generate $1000$ ER graphs with density $p$ and $1000$ SBM graphs with $p_{\rm in}$ and $p_{\rm out}$ such that $p_{\rm in} - p_{\rm out} = \delta$ and the density $(p_{\rm in}+p_{\rm out})/2 = p$. All graphs are generated with $n = 1000$ number of nodes. 

Figure~\ref{sbm:embed} shows the $2$-dimensional embeddings of vector representations $\phi(G)$ of ER graphs with $p = 0.05$ and SBM graphs with various $\delta$ values. The low-dimensional embedding is performed using the Principal Component Analysis (PCA). Specifically, Figure~\ref{sbm:embed}(a) shows the $2$-dimensional embeddings of the  Walk2Vec graph representations. Observe that for large $\delta$ value, the SBM graphs are far away from the ER graphs. As the $\delta$ value decreases, the SBM graphs get closer to the ER graphs. Figure~\ref{sbm:embed}(b) shows the $2$-dimensional embeddings of the Walk2Vec-SC graph representations. Observe that for large $\delta$ value, the SBM graphs form clusters. One can almost see a loop as the $\delta$ value decreases. This implies that the Walk2Vec-SC method could also be useful in estimating the SBM parameter~$\delta$. 

Finally, observe how, especially for the Walk2Vec-SC embeddings, stochastic block instances of the same $\delta$ parameter cluster together in space. This is a very desired property of graph embeddings since it demonstrates robustness and stability of such embeddings in the presence of model randomness or noise. Furthermore, in a similar problem setting, where the goal was to discriminate between graph instances generated by two different model parameters, we would prefer the separation of similar instances into tightly-knit, separable clusters.

\begin{figure}[h!]
\vspace{-0.2in}
\centering
\begin{subfigure}[b]{0.95\linewidth}
\centering
\includegraphics[width=1\linewidth]{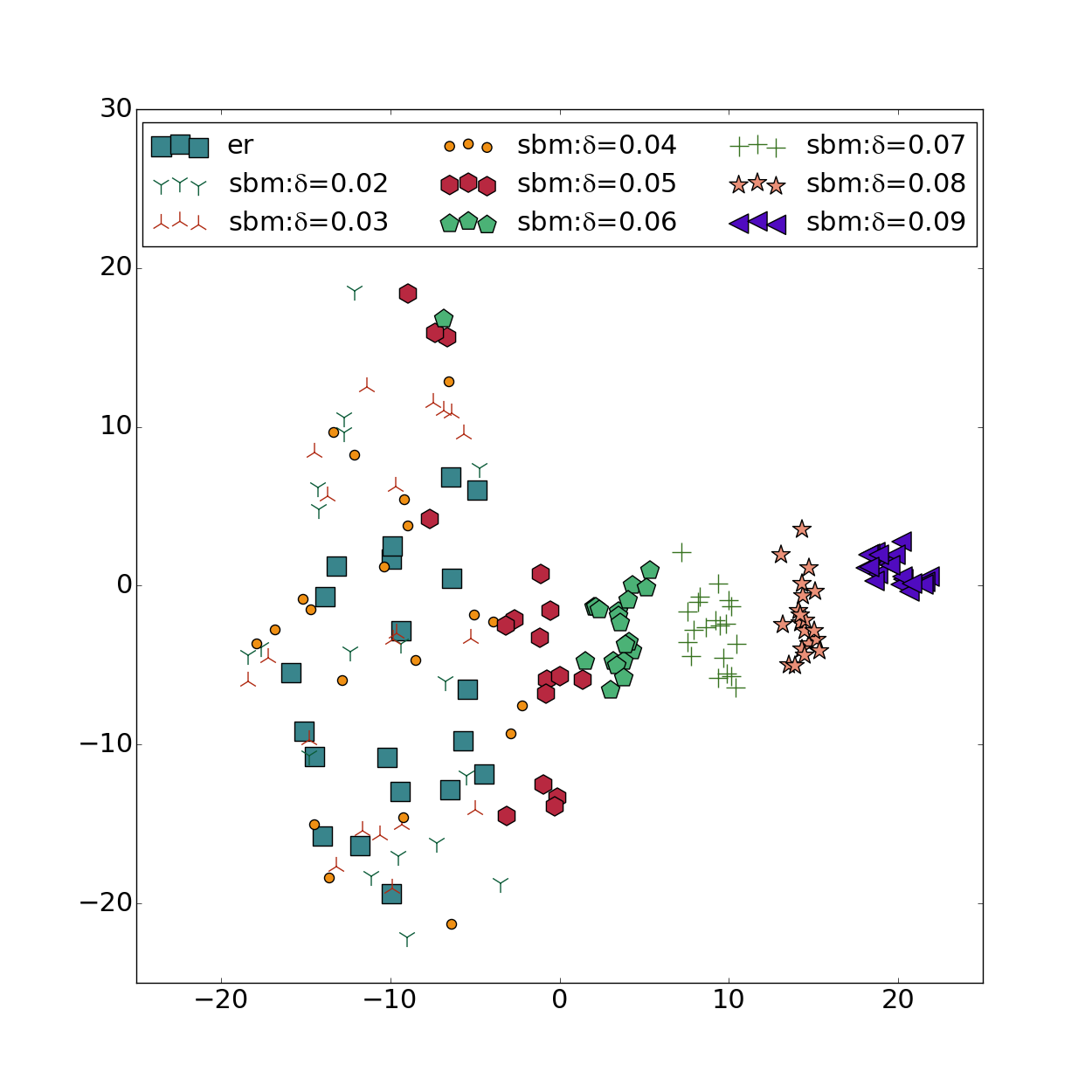}
\vspace{-0.3in}
\caption{Walk2Vec}
\end{subfigure}
\begin{subfigure}[b]{0.95\linewidth}
\centering
\includegraphics[width=1\linewidth]{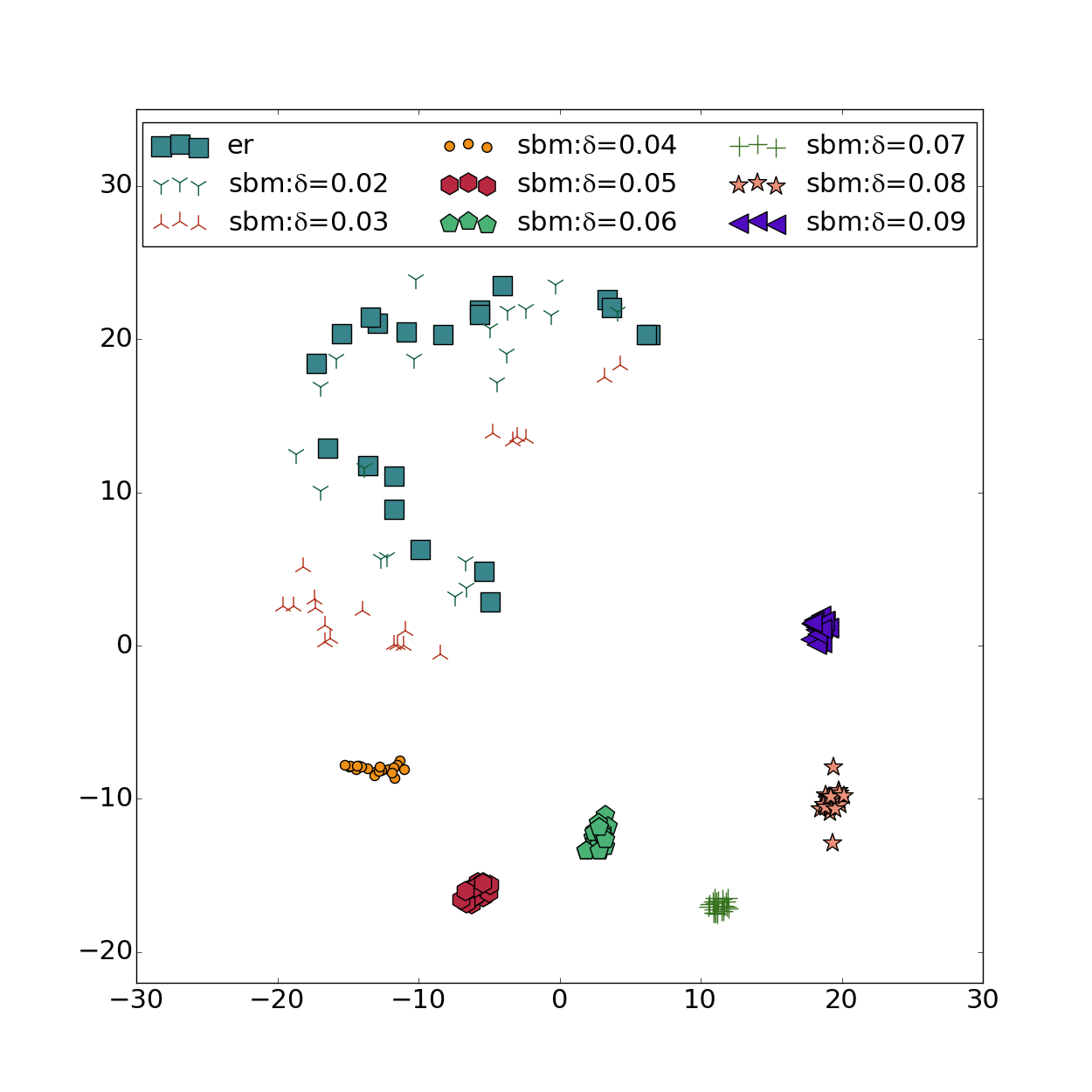}
\vspace{-0.3in}
\caption{Walk2Vec-SC}
\end{subfigure}
\vspace{-0.15in}
\caption{Two-dimensional embeddings of the vector graph representations of ER graphs and SBM graphs. All vector representations are computed using $\tau=15$ random walk steps. }
\label{sbm:embed}
\end{figure}

After mapping each graph into its vector representation $\phi(G)$ using the proposed methods, we then train a random forest classifier on the $500$ vector instances of each model and tested on the $500$ vector instances of each model. 
Fig. \ref{fig:pt} shows the AUC performance on the ER vs. SBM problem using Walk2Vec as $\delta$ increases and for various number of random walk steps. Here $p=0.05$ and $n=1000$. The dashed vertical line represent the theoretical threshold $\delta_{\rm crit.}$ for discriminating ER and SBM models. Observe that the phase transition curve gets sharper as the number of random walk steps $\tau$ increases. For $\tau>10$, increasing $\tau$ has little effect on the performance. For the rest of the experiment, we set $\tau=15$.

The two heat maps in Fig.~\ref{fig:sbm_heat} show the results of the graph classification using Walk2Vec and Walk2Vec-SC, respectively, for various densities $p$ and SBM $\delta$ values. The dark red corresponds to $AUC \approx 1$ and dark blue represent $AUC \approx 0.5$ (i.e., random detection). The dashed line represents the theoretical limit $\delta_{\rm crit.}$. Observe that simulations agree well with the analytical prediction. In particular, the Walk2Vec-SC system exhibits a very sharp phase transition that almost overlaps with the analytical prediction at all densities~$p$.

\begin{figure}[h!]
\centering
\includegraphics[width=0.95\linewidth]{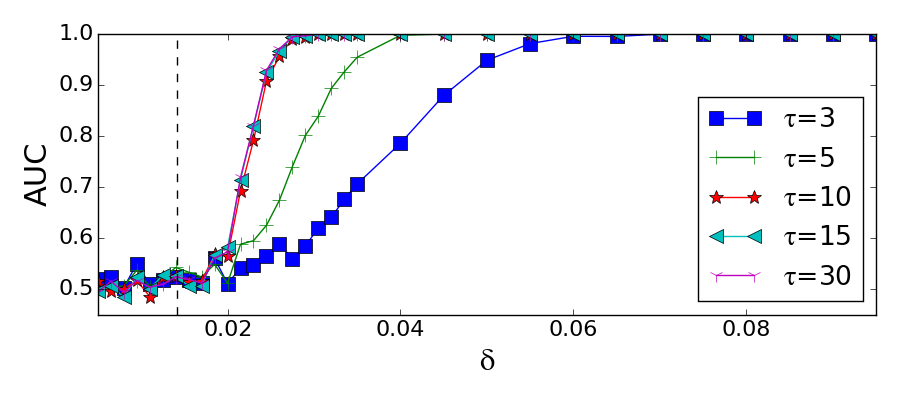}
\smallskip
\vspace{-0.3in}
\caption{AUC performance on the ER vs. SBM problem using Walk2Vec and with varying random walk steps. The dashed line corresponds to the analytical prediction of the phase transition $\delta_{\rm crit.}$.}\label{fig:pt}
\end{figure}

\begin{figure}[h!]
\centering
\begin{subfigure}[b]{1\linewidth}
\vspace{-0.3in}
\centering
\includegraphics[width=1\linewidth]{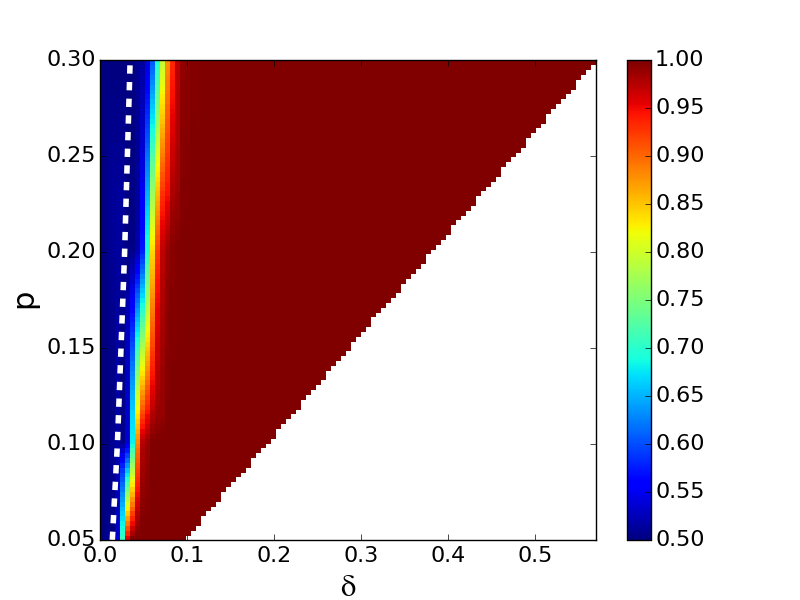}
\caption{Walk2Vec}
\end{subfigure}
\begin{subfigure}[b]{1\linewidth}
\vspace{-0.25in}
\centering
\includegraphics[width=1\linewidth]{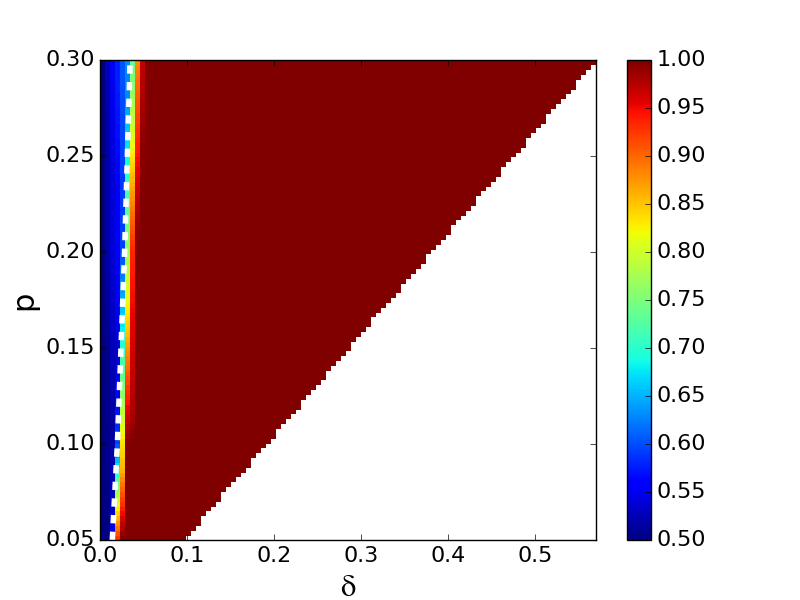}
\vspace{-0.25in}
\caption{Walk2Vec-SC}
\end{subfigure}
\vspace{-0.25in}
\caption{AUC performance on the ER vs. SBM problem. The dashed white line represents the analytical prediction of the phase transition $\delta_{\rm crit.}$ for various densities $p$.}\label{fig:sbm_heat}
\end{figure}
\vspace{-0.1in}

\label{sec:er_sbm} 

\vspace{-0.1in}
%
%
\subsection{Planted Clique Problem}\label{sec:planted}
In this section, we consider the related problem of distinguishing between an ER graph and an ER
graph with a planted clique of size $k$. Let $\beta = k/\sqrt{n}$ where $n=|V|$. The classification problem gets harder as the size of the clique $k$ becomes smaller. As shown in~\cite{nadakuditi2012hard}, the limit of detecting a planted clique graph from a ER graph is 
\begin{align}
\beta_{\rm crit.} = \sqrt{\frac{p}{1-p}}\enspace,
\end{align}
where $p$ is the connection probability of ER graphs.

To generate a graph with a planted clique, we first generate an ER graph with density $p$. Then we randomly select $k$ nodes from the graph and connect all pairs of distinct nodes in the selected node set. For each pair of $(p,\beta)$, we generate $1000$ ER graphs with density $p$ and $1000$ ER graphs of the same density and with a planted clique of size $k$ such that $k/\sqrt{n} = \beta$. All graphs are generated with $n=1000$ number of nodes.

Figure~\ref{planted:embed} shows the $2$-dimensional embeddings of the vector representations $\phi(G)$ of ER graphs with $p = 0.5$ and ER graphs with planted cliques. All graphs are of the size $n=1000$. The embedding is performed using PCA. Observe from both Figure~\ref{planted:embed}(a) and \ref{planted:embed}(a) that graphs with large planted clique (i.e., large $\beta$ value) are further away from the ER graphs. As $\beta$ decreases, these graphs move closer to the ER graphs. 

\begin{figure}[h!]
\centering
\begin{subfigure}[b]{0.95\linewidth}
\centering
\includegraphics[width=1\linewidth]{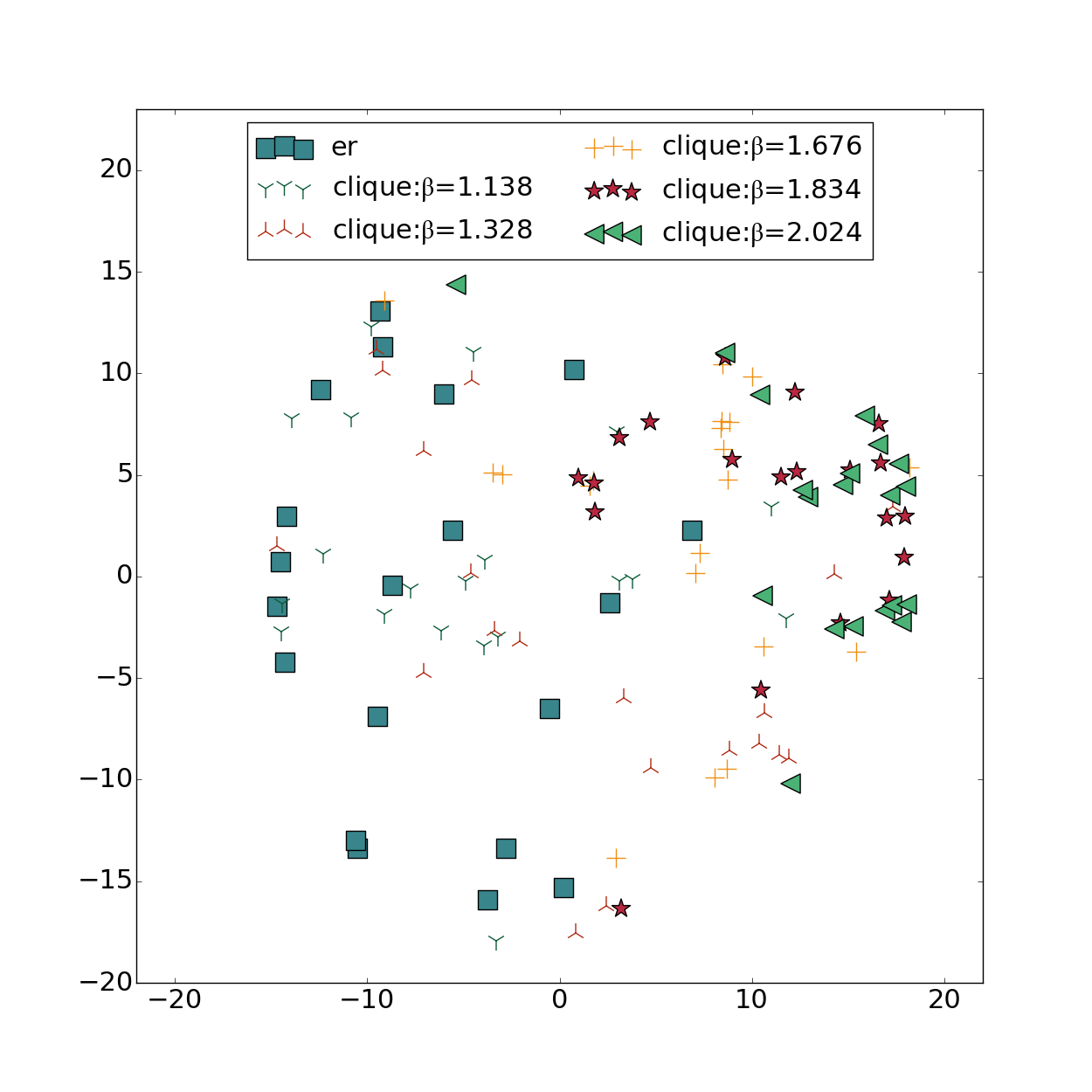}
\vspace{-0.3in}
\caption{Walk2Vec}
\end{subfigure}
\begin{subfigure}[b]{0.95\linewidth}
\centering
\vspace{-0.03in}
\includegraphics[width=1\linewidth]{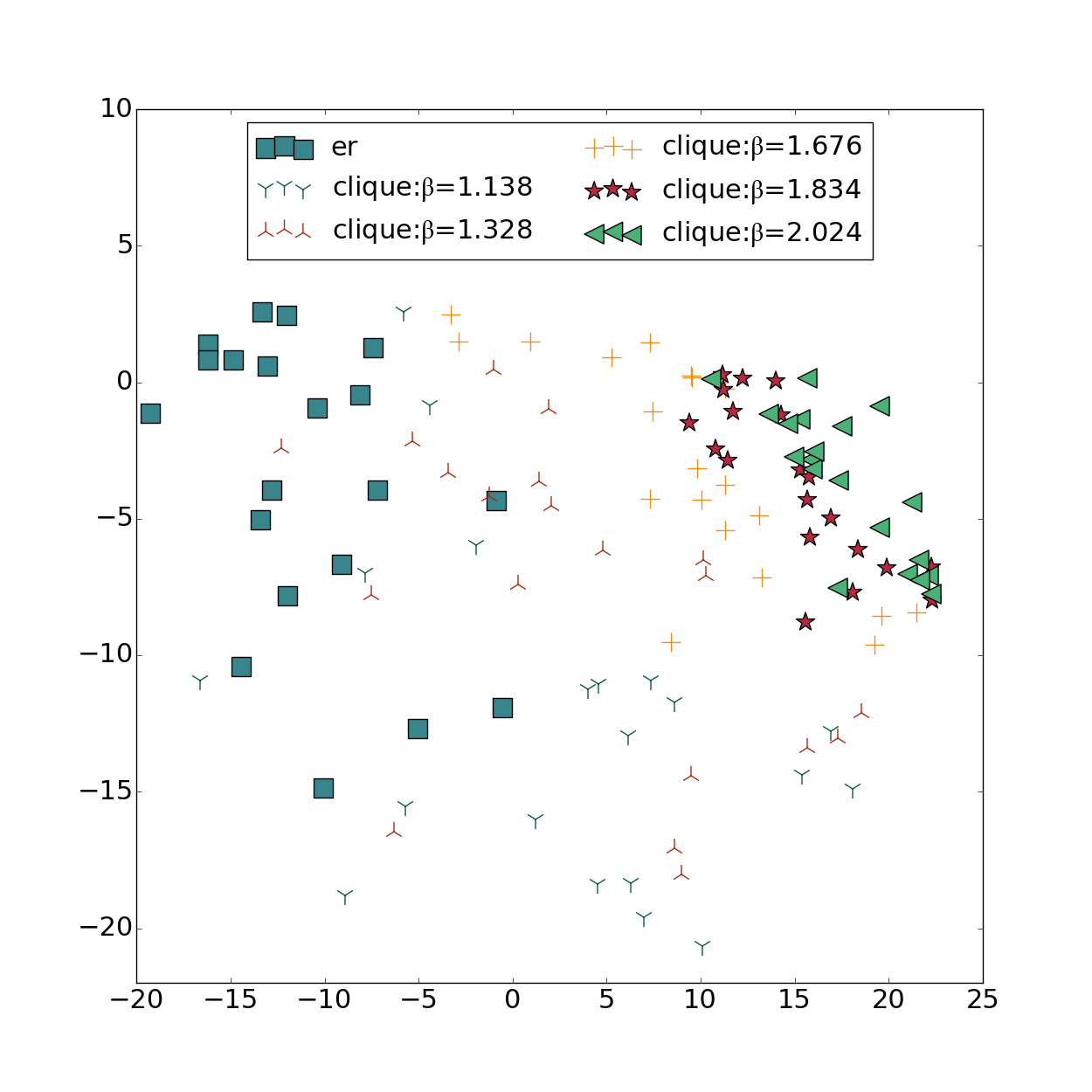}
\vspace{-0.3in}
\caption{Walk2Vec-SC}
\end{subfigure}
\caption{Two-dimensional embeddings of the vector graph representations of ER graphs and ER graphs with planted cliques.}
\label{planted:embed}
\end{figure}

After mapping each graph into its vector representation $\phi(G)$ using the proposed methods, we then train a random forest classifier on the 500 vector instances of each model and tested on the 500 vector instances of each model.
Fig.~\ref{fig:clique_heat} shows AUC performance of the planted clique problem for various graph densities $p$ and clique parameter values $\beta$. The dark red correspond to $AUC \approx 1$ and the dark blue corresponds to $AUC \approx 0.5$. The dashed line represent the theoretical limit for clique detection. Observe from Fig.~\ref{fig:clique_heat} that both Walk2Vec and Walk2Vec-SC systems perform well. The agreement between the theoretical limit and the simulations is excellent for $n=1000$. As $n$ increases, the transition is expected to get sharper.

\begin{figure}[h!]
\centering
\begin{subfigure}[b]{1\linewidth}
\centering
\includegraphics[width=1\linewidth]{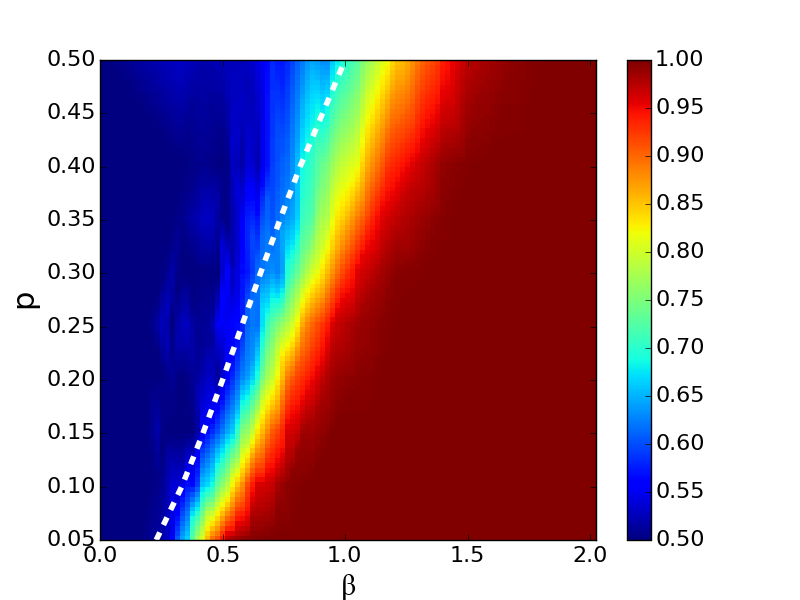}
\caption{Walk2Vec}
\end{subfigure}
\begin{subfigure}[b]{1\linewidth}
\centering
\includegraphics[width=1\linewidth]{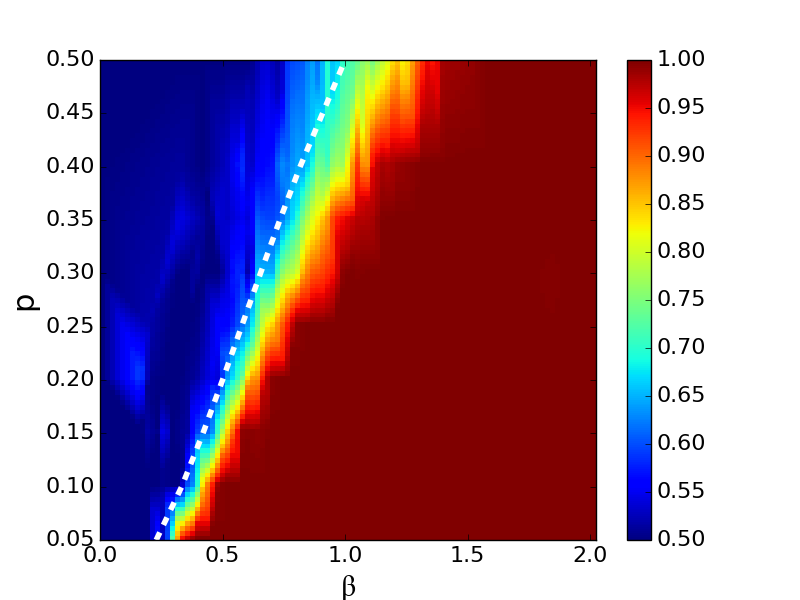}
\caption{Walk2Vec-SC}
\end{subfigure}
\caption{AUC performance of the planted clique problem. The dashed line represents the analytical phase transition prediction $\beta_{\rm crit.}$.}\label{fig:clique_heat}
\end{figure}

\label{sec:pc}

\subsection{Performance Comparison}
We compare the performance of the Walk2Vec and Walk2Vec-SC embeddings with that of the graph topological feature embeddings discussed in~\cite{caceres2016}. We consider the following $26$ topological features: degree centrality (1-4), betweenness centrality (5-8), closeness centrality (9-12), clustering coefficient (13-16), diameter (17), radius (18), triad count (19-22), average shortest path length (23-26). Similarly to the experimental setup in~\cite{caceres2016}, if one feature is assigned four numbers, this means we considered the maximum, the minimum, the average and the standard deviation over each node in the graph. We train a random forest classifier and use the learned model to classify a collection of unlabeled graph instances.

\begin{figure}[h!]
\centering
\includegraphics[width=0.95\linewidth]{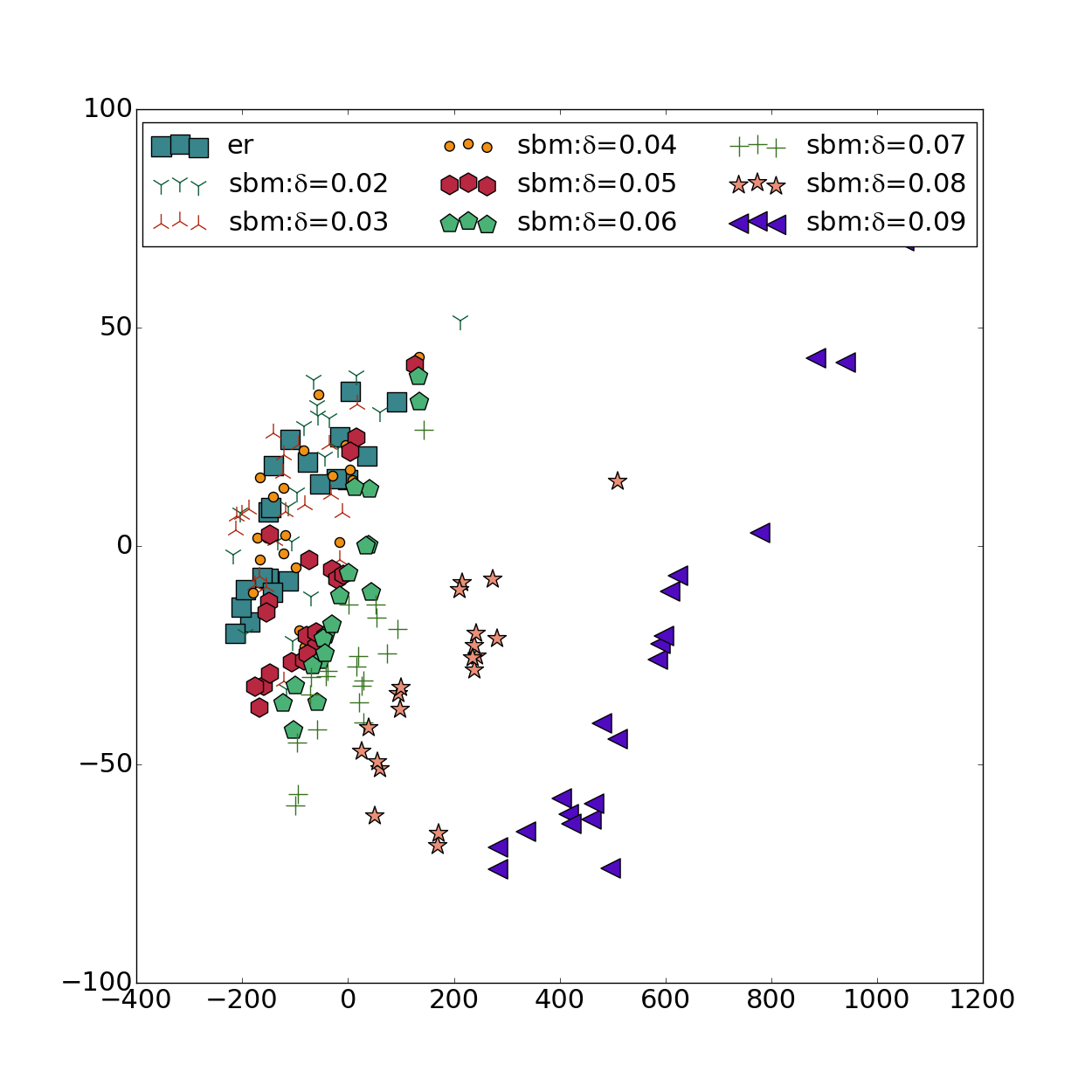}
\smallskip
\vspace{-0.2in}
\caption{Two-dimensional embeddings of the topological features of ER graphs and SBM graphs}\label{fig:top_sbm}
\end{figure}

\begin{table}[h!]
\centering
\caption{Performance comparison on the ER vs. SBM problem. Here $n=1000$ and $p = 0.05$.}
\small
\begin{tabular}{ | l | c  |c | c|}
  \hline
  $\delta$ & Topological Feats. & Walk2Vec & Walk2Vec-SC\\ 
\hlinewd{1pt}
  0.005 & 0.51& 0.48& 0.52\\
  0.008 & 0.49& 0.52& 0.48\\
  0.011 & 0.52& 0.52& 0.51\\
  0.014 & 0.56& 0.50& 0.61\\
\hdashline
  0.017 & 0.68& 0.52& 0.78\\
  0.02 & 0.82& 0.56& 0.95\\
  0.023 & 0.92& 0.72& 0.998\\
  0.026 & 0.98& 0.90& 1.0\\
  0.03 & 0.999& 0.99& 1.0\\
  0.04 & 1.0& 0.999& 1.0\\
  0.05 & 1.0& 1.0& 1.0\\
  0.06 & 1.0& 1.0& 1.0\\
  0.07 & 1.0& 1.0& 1.0\\
  0.08 & 1.0& 1.0& 1.0\\
  \hline
\end{tabular}
\label{table:sbm}
\end{table}

\begin{figure}[h!]
\centering
\includegraphics[width=0.95\linewidth]{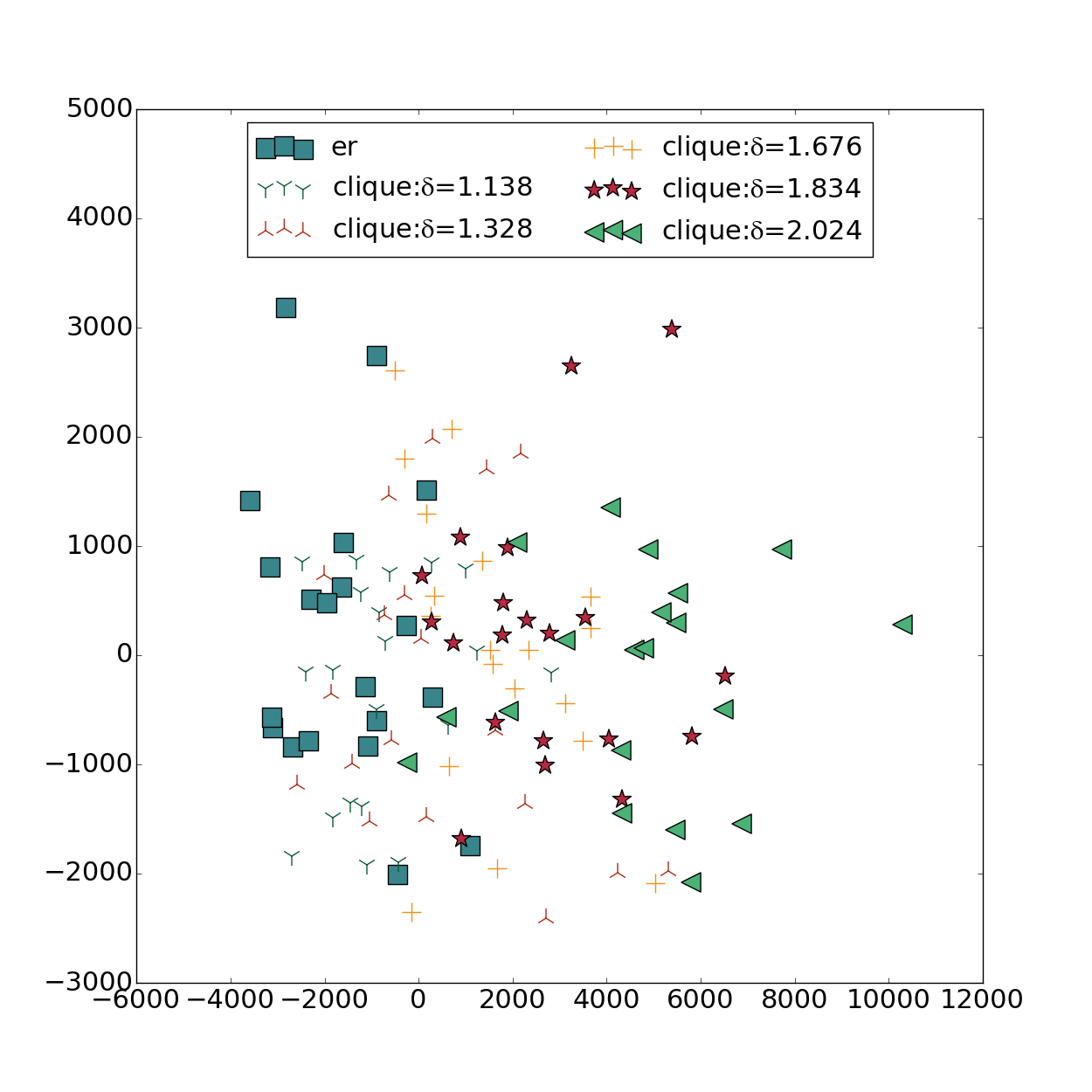}
\smallskip
\vspace{-0.2in}
\caption{Two-dimensional embeddings of the topological features of ER graphs with and without cliques}\label{fig:top_clique}
\end{figure}

\begin{table}[h!]
\centering
\caption{Performance comparison on the planted clique problem. Here $n = 1000$ and $p = 0.5$ for all ER graphs.}
\small
\begin{tabular}{ | l |c| c  |c | c|}
  \hline
  k &$\beta$ & Topological & Walk2Vec & Walk2Vec-SC\\   
  & & Feats. &  &  \\ 
\hlinewd{1pt}
  10&0.316 & 0.48& 0.51& 0.50\\
  21&0.664 & 0.52& 0.54& 0.60\\
  31&0.980 & 0.71& 0.62& 0.71\\
\hdashline
  33&1.044 & 0.77& 0.67& 0.77\\
  36&1.138 & 0.87& 0.70& 0.84\\
  39&1.233 & 0.94& 0.77& 0.92\\
  42&1.328 & 0.98& 0.81 & 0.97\\
  47&1.486 & 0.999& 0.90& 0.998\\
  53&1.676 & 1.0& 0.97& 1.0\\
  58&1.834 & 1.0& 0.99& 1.0\\
  64&2.024 & 1.0& 1.0& 1.0\\
  \hline
\end{tabular}
\label{table:planted}
\end{table}

Table~\ref{table:sbm} shows the AUC performance comparison of topological features, Walk2Vec and Walk2Vec-SC on ER graphs and SBM graphs for different $\delta$ values. All graphs are generated with $n = 1000$ number of nodes and $p = 0.05$. The horizontal dashed line in the table represents the location of the theoretical limit for discriminating ER graphs and SBM graphs. Observe that the performance is comparable between topological features and Walk2Vec representation, while Walk2Vec representation for a graph instance is considerable cheaper than the topological features. Additionally, the Walk2Vec-SC performs the best out of the three. The phase transition is very sharp around the threshold, i.e., the dashed line.

Table~\ref{table:planted} shows the AUC performance comparison 
of the three methods on ER graphs with and without planted cliques. All ER graphs are generated with $n=1000$ and $p = 0.5$. Note that the value of $\beta$ corresponds to the size of the planted clique. The horizontal dashed line represents the location of the theoretical value of the phase transition. Observe from Table~\ref{table:planted} that the performance is comparable between the topological features and Walk2Vec-SC representations. Both exhibit very sharp transition around the threshold. On the other hand, the phase transition of Walk2Vec is not as sharp. This method could be useful for detecting large cliques as it is the most computationally efficient method of graph vector representation among the three methods. 

A notable difference between the random walk representations generated by our two methods and the topological representation of~\cite{caceres2016} is the quality of clustering of different graph instances of the same model parameter. As illustrated in Figure~\ref{fig:top_sbm} and ~\ref{fig:top_clique}, topological feature embeddings appear to be much more sensitive to variations due to model randomness, and therefore, less robust in capturing the inherent structural similarity of instances generated by the same model parameter. By contrast, the Walk2Vec and Walk2Vec-SC representations appear do a much better job in smoothing out randomness effects. In realistic scenarios, we expect Walk2Vec and Walk2Vec-SC to be much more robust in handling inherent noise in observed graph instances.

In addition, both Walk2Vec and Walk2Vec-SC are scalable to large graphs. The Walk2Vec computes the random walk features on selected nodes and stack the feature vectors. For sparse graphs, the Walk2Vec's computation complexity  is $O(n)$ where $n$ is the number of nodes in the graph. For Walk2Vec-SC, since the random walk feature is computed for every node in the graph and the dictionary learing is linear, the computation complexity for Walk2Vec-SC is thus $O(n^2)$. Note that one can parallelize the computation of the random walk feature for each node, thus the time it takes to compute the Walk2Vec-SC representation can be shorten significantly, making it also scalable to large graphs. On the other hand, the compuation of topological features is dominated by the average shortest path length, whose computation complexity is $O(n^2\log n)$ for directed graphs and $O(n^2)$ for undirected graphs, while Walk2Vec and Walk2Vec-SC apply to both weighted and unweighted graphs.

\section{Conclusion}

In this paper, we propose Walk2Vec, a novel approach that uses random
walks for learning robust graph representations. Our method learns discriminating features by leveraging different mechanisms to initiate random walks and by correlating temporal dependencies between random walk steps. These representations are
invariant under graph isomorphism and graph size. Experimental results on two challenging graph model selection problems show classification performance that closely matches known theoretical limits, implying that the Walk2Vec approach can map graphs into meaningful representations. Furthermore, these learned representations are robust to inherent randomness or noise in the data generation process, while simple and scalable to compute for large graphs.

%
%
\bibliographystyle{unsrt}\scriptsize
\bibliography{recog,graph,entity,related}

\end{document}